\documentclass[10pt]{article}

\usepackage[margin=1in]{geometry}
\usepackage{natbib}
\usepackage{parskip}

\usepackage{algorithm, algorithmicx, algpseudocode}
\usepackage{amssymb}
\usepackage{amsmath}
\usepackage{graphicx}
\usepackage{mathtools}
\usepackage{needspace} 
\usepackage{multicol}
\usepackage{listings}
\usepackage{amsthm}
\usepackage[shortlabels]{enumitem}
\usepackage[italicdiff]{physics}
\usepackage{cancel}
\usepackage[dvipsnames]{xcolor}
\usepackage{verbatim}
\usepackage{comment}
\usepackage{array}
\usepackage{tikz-cd}
\usepackage{import}
\usepackage{booktabs}
\usepackage{comment}
\usepackage{array}
\usepackage[colorlinks, linkcolor=blue, citecolor=ForestGreen]{hyperref}
\usepackage{forest}
\usepackage{float}
\usepackage{xspace}
\usepackage{aliascnt}
\usepackage{wrapfig}
\usepackage{thm-restate}
\usepackage{adjustbox}
\usepackage{bm}
\usepackage{bbm}
\usepackage{complexity}
\usepackage{makecell}
\usepackage{nicefrac}
\usetikzlibrary{calc, positioning, external, arrows.meta}

\usepackage{cleveref}

\allowdisplaybreaks

\let\E\relax

\DeclareMathOperator*{\argmax}{argmax}
\DeclareMathOperator*{\E}{\mathbb E}

\let\cite\citep

\newfloat{protocol}{tbp}{lop}
\floatname{protocol}{Protocol}

\usepackage{pgfplots}
\usepgfplotslibrary{groupplots}
\pgfplotsset{compat=1.18}

\newtheorem{theorem}{Theorem}
\numberwithin{theorem}{section}
\newtheorem*{theorem*}{Theorem}

\newtheorem{lemma}[theorem]{Lemma}

\theoremstyle{definition}
\newtheorem{definition}[theorem]{Definition}
\theoremstyle{definition}
\newtheorem{example}{Example}[section]

\usepackage{amsmath,amsfonts,bm}

\def\1{\bm{1}}

\def\vt{{\bm{t}}}

\DeclareMathAlphabet{\mathsfit}{\encodingdefault}{\sfdefault}{m}{sl}
\SetMathAlphabet{\mathsfit}{bold}{\encodingdefault}{\sfdefault}{bx}{n}

\usepackage{caption}
\makeatletter
\def\thanks#1{\protected@xdef\@thanks{\@thanks
        \protect\footnotetext{#1}}}
\makeatother

\title{Online Resource Allocation With General Constraints}
\author{
    Eleonora Fidelia Chiefari \\
    \texttt{eleonorafidelia.chiefari@polimi.it} \\
    Politecnico di Milano 
\and
    Francesco Emanuele Stradi \\
	\texttt{francescoemanuele.stradi@polimi.it} \\
	Politecnico di Milano 
\and
	Matteo Castiglioni\\
	\texttt{matteo.castiglioni@polimi.it} \\
	Politecnico di Milano 
\and 
	Alberto Marchesi\\
	\texttt{alberto.marchesi@polimi.it} \\
	Politecnico di Milano 
}
\date{\today}

\begin{document}

\maketitle

\begin{abstract}
	\emph{Online resource allocation} (ORA) is a fundamental framework for sequential decision-making problems under \emph{budget constraints}, with applications ranging from online advertising to revenue management. In this work, we study a broader setting that includes \emph{both} budget constraints and \emph{general constraints}, extending the classical budget-only model. This extension is essential for modeling critical economic requirements, such as Return-on-Investment (ROI) constraints.
	We develop an algorithm that achieves best-of-both-world guarantees within this generalized framework. In particular, against a dynamic benchmark, our algorithm achieves \(\widetilde{\mathcal O}(\sqrt{T})\) regret in the \emph{stochastic} regime and \(\alpha\)-regret of order \(\widetilde{\mathcal O}(\sqrt{T})\) in the \emph{adversarial} regime, where \(\alpha\) depends on the feasibility margin of the corresponding offline problem. At the same time, our algorithm guarantees strict satisfaction of the budget constraints and \(\widetilde{\mathcal O}(\sqrt{T})\) cumulative violation for the general ones.
	From a technical perspective, introducing general constraints alongside budgets precludes the use of standard budget-focus methods. While budget methods rely on a zero-consumption ``safe'' action to ensure feasibility, general constraints are much less ``aligned'' towards feasibility. We overcome these difficulties with a new analysis that exploits \emph{weak adaptivity} to get boundedness of the Lagrangian multipliers and best-of-both-world guarantees.
\end{abstract}

\newpage
\tableofcontents
\newpage

\section{Introduction} 

\emph{Online resource allocation} (ORA) serves as a fundamental framework for sequential decision-making under \emph{budget constraints}. In these settings, requests arrive sequentially, and, for each request, the decision maker must select an appropriate action. A defining characteristic of ORA, which distinguishes it from many online learning paradigms, is that the learner possesses exact knowledge of the current round \emph{before} picking an action: upon a request's arrival, the resource consumption and reward associated with each available action are fully known. Hence, the challenge stems from the uncertainty about future requests, requiring a strategy that balances immediate rewards against the preservation of resources for future arrivals. 

Such ORA problems arise in various real-world scenarios, such as, \emph{e.g.}, online advertising~\citep{mehta}, revenue management~\citep{talluri2004theory, ball2009toward}, and online auctions~\citep{zhou,balseiro2019learning}. 
A wide range of works have addressed ORA problems in various stochastic settings~\citep{devanur2009the,feldman2010,KesselheimRTV13,agrawal2014dynamic,devanur_nearoptimal,li2020,sun2022nearoptimalprimaldualalgorithmsquantitybased}. Only a few works have gone beyond these settings: some consider non-stationary models~\citep{Ciocan2011DynamicAP, esfandiari}, while others address adversarial ones, albeit under specific structural assumptions~\citep{mehta, buchbinder, balseiro2019learning, stradino}.
\citet{balseiro2023best} are the first to consider general ORA problems in both stochastic and adversarial regimes by designing a \emph{best-of-both-world} algorithm. Their algorithm employs a dual-based approach that simultaneously achieves \( \widetilde{\mathcal{O}}(\sqrt{T}) \) regret in the stochastic regime and \( \alpha \)-regret of order \( \widetilde{\mathcal{O}}(\sqrt{T}) \) in the adversarial setting, where \( \alpha \) is an instance-dependent constant parameter related to the per-round budget.

Despite these progresses, most existing results focus exclusively on budget constraints. A notable exception is the work of \cite{feng2023onlinebiddingalgorithmsreturnonspend}, which, however, focuses only on Return-on-Investment (ROI) constraints within stochastic settings.
Despite that, many real-world applications involve additional \emph{long-term constraints}---such as fairness, risk, or performance requirements---that must be satisfied cumulatively over time. A fundamental difference between general and budget constraints is that while budget constraints come with an \emph{a priori} known feasibility margin (the per-round budget), general constraints often have an \emph{unknown} feasibility margin. Furthermore, such constraints are inherently less ``aligned'' with one another. For instance, an action that maximizes ROI may simultaneously violate fairness. This tension requires algorithms to allow violations in some rounds to be later compensated by strict satisfaction in others, thereby forcing learning algorithms to tolerate temporary violations. These dynamics introduce additional challenges that prevent results for budget-constrained settings from being directly applied to problems involving general constraints.

\subsection{Original Contributions}

In this work, we study generalized ORA problems that incorporate \emph{both budget and general long-term constraints}. To the best of our knowledge, we propose the first algorithm that attains \emph{best-of-both-world} guarantees in this setting. 
Notably, while the performance of our algorithm in adversarial environments depends on the Slater's parameter $\rho$---which encodes the feasibility margin of the problem---the algorithm operates without prior knowledge of this parameter.
Our algorithm achieves the following performance guarantees: 
\begin{itemize}[leftmargin=0.5cm]
	\item \( \widetilde{\mathcal{O}}(\sqrt{T}) \) cumulative expected regret in the \emph{stochastic} setting, where requests are drawn i.i.d.\ across rounds from an unknown fixed probability distribution; 
	\item \( \widetilde{\mathcal{O}}(\sqrt{T}) \) cumulative \( \alpha \)-regret in the \emph{adversarial} setting, where requests are chosen arbitrarily in advance by an adversary, with \( \alpha := \nicefrac{\rho}{1+\rho} \) an instance-dependent constant.
\end{itemize}
At the same time, regardless of the underlying environment, the algorithm: 
\begin{itemize}[leftmargin=0.5cm]
	\item strictly satisfies all the budget constraints;
	\item guarantees \( \widetilde{\mathcal{O}}(\sqrt{T}) \) cumulative violation of the general constraints.
\end{itemize}

Our algorithm employs a \emph{dual-based} approach, wherein Lagrangian multipliers are updated via an \emph{online gradient descent} (OGD) scheme and primal actions are selected greedily in order to maximize the instantaneous Lagrangian.
The central challenge in our analysis lies in bounding the magnitude of these Lagrangian multipliers throughout the primal-dual dynamics. Crucially, we do not impose an external upper bound on the dual variable, as such a bound depends on the unknown Slater's parameter $\rho$.
Instead, we show that the multipliers are naturally bounded under the learning dynamics.

The crucial difference with respect to the classical budget-only setting~\citep{balseiro2023best} stems from the absence of a point-wise feasibility-maximizing action, a.k.a.\ the ``void'' action.
In our more general setting, constraints allow for cross-constraint compensation, which can prevent the violations from vanishing and the multipliers from ceasing to grow. Indeed, since the constraints might not be aligned, it could be the case that multipliers are large, but the most profitable decision for the primal is to keep violating a constraint to satisfy another one by a huge margin.
Despite that, exploiting the \emph{weak adaptivity} of OGD, we are still able to prove that, when the multipliers exceed a certain threshold, the no-regret guarantees of the OGD dual update prevent the primal, which plays greedily, from indefinitely exploiting this cross-constraint compensation.
This allows us to prove that the norm of the multipliers cannot exceed a threshold inversely proportional to the feasibility margin $\rho$ with high probability, completely bypassing the need for a point-wise feasibility-maximizing action.

\subsection{Related Works}
A wide range of works on the ORA problem focuses on stochastic input models. Early contributions \citep{devanur2009the, feldman2010} introduced algorithms that learn dual prices from an initial sample and exploit them. Although conceptually simple, these algorithms attained suboptimal guarantees. Subsequent works \citep{agrawal_nearoptimal, KesselheimRTV13, devanur_nearoptimal} improved performance by sequentially solving linear programming relaxations, achieving $\widetilde{\mathcal O}(\sqrt{T})$ regret. 
A growing line of research aims to design computationally efficient algorithms, avoiding heavy optimization steps. Recent works \citep{li2020, sun2022nearoptimalprimaldualalgorithmsquantitybased} employ primal-dual or gradient-style updates, attaining strong theoretical guarantees although relying on specific problem structures or assumptions.

More recent works \citep{Ciocan2011DynamicAP, esfandiari} explore non-stationary models, filling the gap between stochastic and adversarial ones. For what concerns the adversarial setting, early works \citep{mehta, buchbinder} establish optimal constant-factor guarantees for the specific cases of online matching and AdWords problem. These results assume linear relationship between reward and consumption and cannot be directly applied to more general allocation settings.
The first work providing guarantees in both the stochastic and the adversarial setting is \citet{balseiro2023best}. The authors develop a dual-based algorithm that achieves simultaneously $\widetilde{\mathcal O}(\sqrt{T})$ regret in the stochastic setting, $\widetilde{\mathcal{O}}(\sqrt{T})$ $\alpha$-regret in the adversarial setting and in the non-stationary setting.

In the specific domain of online bidding, \citet{balseiro2019learning} provide an adaptive pacing strategy achieving $\widetilde{\mathcal O}(\sqrt{T})$ regret bound in the stochastic setting, under suitable assumptions. In the adversarial setting, it establish a fundamental limitation, showing that no algorithm can achieve a competitive ratio better than $\nicefrac{v}{\rho}$, where $v$ is the uniform upper bound on the advertiser's value. Similarly, \citet{feng2023onlinebiddingalgorithmsreturnonspend} consider an online bidding problem under a RoS constraint, where the goal is to maximize the total value obtained from ad placements while ensuring that the ratio between value and expenditure exceeds a given threshold. Their analysis relies on the structure of the RoS constraint to obtain optimal regret guarantees. Alternative approaches have also studied competitive guarantees under budget constraints. In particular, \citet{zhou} design online algorithms whose performance depends on the structure of value-to-weight ratios of the incoming items. 

A closely related line of research focuses on online decision-making subject to long-term constraints. Early works, such as \citep{mahdavi2012tradingregretefficiencyonline, jenatton2015adaptivealgorithmsonlineconvex} consider fixed constraints, achieving sublinear regret while controlling cumulative violation. Subsequent works, such as \citep{Mannor, yu2017onlineconvexoptimizationstochastic, pmlr-v70-sun17a}, extend these results to time-varying constraints. The \emph{Bandits with Knapsacks} (BwK) framework,introduced by \citet{Badanidiyuru_2013}, focuses on budget constraints. Several works, such as \citep{agrawal_globalconvex} have improved the regret guarantees in the stochastic setting. More recently, primal-dual algorithms for BwK have been extended to adversarial and non-stationary environments \citep{immorlica2023adversarialbanditsknapsacks, kesselheim2020, castiglioni2022online}. Most closely related to our research are \cite{bernasconi2024noregretenoughbanditsgeneral,castiglioni2024online}, which generalize the BwK framework to more general constraints. Also in their case, the analysis is built upon the weak adaptivity of the primal and dual regret minimizers.

\section{Preliminaries}\label{sec:preliminaries}

We study \emph{online resource allocation} (ORA) problems in which a decision maker (hereafter, agent) selects actions over a finite horizon of \( T \) rounds. We denote the action space of the agent by \( \mathcal{X} \subseteq \mathbb{R}^K \). At each round $t \in [T]$,\footnote{Throughout this paper, we denote by $[a] := \{1,\ldots,a \}$ the set of the first $a \in \mathbb{N}$ natural numbers.} the agent observes an input tuple $\gamma_t = (f_t, \bm{g}_t, \bm{h}_t)$, where:
\begin{itemize}[leftmargin=0.7cm]
	\item $f_t: \mathcal{X} \to [0,1]$ is the reward function;
	\item $\bm{g}_t : \mathcal{X} \to [-1,1]^m$, whose $i$-th component $g_{t,i} : \mathcal{X} \to [-1,1]$ encodes the cost function associated with constraint $i \in [m]$; 
	\item $\bm{h}_{t} : \mathcal{X} \to [0,1]^n$, whose $j$-th component $h_{t,j} : \mathcal{X} \to [0,1]$ encodes the consumption function associated with resource $j \in [n]$.
\end{itemize}
The input functions are fully revealed to the agent before they make any decision. After observing the input, the agent selects an action $x_t \in \mathcal{X}$, which consumes $h_{t,j}(x_t)$ of each resource $j \in [n]$, produces a reward $f_t(x_t)$ and incurs a constraint cost $g_{t,i}(x_t)$ for each $i \in [m]$, where strictly positive values of $g_{t,i}(x_t)$ indicate a violation of constraint $i$, while $g_{t,i}(x_t)\leq 0$ means that constraint $i$ is satisfied.
The input tuples may be either \emph{stochastic}, \emph{i.e.}, generated i.i.d.\ across rounds according to a fixed probability distribution, or \emph{adversarial}, \emph{i.e.}, chosen arbitrarily by an adversary. Throughout the paper, we denote by $\Gamma$ the set of all the possible input tuples $\gamma = (f,\bm{g}, \bm{h})$.

The agent has an overall budget $B_j=\Omega(T)$ limiting the cumulative consumption of each resource $j \in [n]$ over the time horizon. We denote by $\bm{\beta} \in \mathbb{R}^n$ the vector representing the per-iteration budget for each resource, where $\beta_j > 0$ is defined as $B_j/T$. 
As it is common in the online resource allocation literature (see, \emph{e.g.},~\citep{balseiro2023best}), we assume that there exists a \textit{void action} $\varnothing \in \mathcal{X}$ such that, for any input $(f, \bm{g}, \bm{h}) \in \Gamma$, it holds $f(\varnothing) = 0$, $g_i(\varnothing) = 0$ for $i \in [m]$ and $h_j(\varnothing) = 0 $ for $j \in [n]$.

Given a specific sequence of inputs $\bm{\gamma} = (\gamma_1, ..., \gamma_T),$ where $\gamma_t = (f_t, \bm{g}_t, \bm{h}_t)$ for all $t$, we define the \emph{cumulative reward} of the agent over the time horizon $T$ as: 
\[\texttt{Rew}(T, \bm{\gamma}):= \sum_{t \in [T]} f_t(x_t).\]
Moreover, we define the \emph{cumulative constraint violation} over the time horizon $T$ for the general, i.e., non-budget, constraints as: 
\[V_T  := \max_{i \in [m]} \sum_{t \in [T]}g_{t,i}(x_t).\]
The goal of the agent is to maximize $\texttt{Rew}(T, \bm{\gamma})$, while ensuring: 
\begin{itemize}
	\item $V_T$ grows sublinearly in $T$, namely, it is $o(T)$; 
	\item $\sum_{t \in [T]}h_{t,j}(x_t) \leq \beta_j T$ for all $j \in [n]$, \emph{i.e.}, the budget constraints are strictly enforced.
\end{itemize}

\subsection{Stochastic Setting} \label{p_s}

In the stochastic setting, the inputs are drawn i.i.d.\ across rounds from a fixed but unknown probability distribution $\mathcal{P}$ supported on $\Gamma$. Formally, at each round $t \in [T]$, it holds $\gamma_t = (f_t, \bm{g}_t, \bm{h}_t) \sim \mathcal{P}$. We denote an input sequence by $\bm{\gamma} = (\gamma_1, ..., \gamma_T) \in \Gamma^T$ and compactly write $\bm{\gamma} \sim \mathcal{P}^T$ to indicate that each $\gamma_t$ is drawn i.i.d.\ from $\mathcal{P}$. Then, we let $\text{OPT}_\text{stoc} \coloneqq \E_{\bm{\gamma} \sim \mathcal{P}^T} [\text{OPT}(\bm{\gamma})]$ be the dynamic-optimum in the stochastic setting, where $\text{OPT}(\bm{\gamma})$ is defined as:
\begin{equation}
	\text{OPT}(\bm{\gamma}):=\begin{cases}
		\max_{(x_1,...,x_T) \in \mathcal{X}^T}\footnotemark &  \sum_{t \in [T]} f_t(x_t)\\
		\,\,\, \textnormal{s.t.} & \sum_{t \in [T]}  g_{t,i}(x_t) \leq 0 \quad \forall i \in [m] \nonumber\\
		\,\,\,  & \sum_{t \in [T]}  h_{t,j}(x_t) \leq \beta_j T \quad \forall j \in [n].
	\end{cases}
\end{equation}
\footnotetext{All the results presented in this work can be easily extended to baselines defined using
	$
	\sup_{(x_1,\ldots,x_T) \in \mathcal{X}^T}.
	$}
We define a key quantity, known as the \emph{Slater's parameter}, related to the feasibility of the problem: 
\begin{align*}
	\rho_\text{stoc} 
	= - \min_{\pi : \Gamma \to \mathcal{X}} \max \left\{ \max_{i \in [m]} \mathbb{E}_{\gamma \sim \mathcal{P}} \left[ g_{i}(\pi(\gamma)) \right], \max_{j \in [n]} \mathbb{E}_{\gamma \sim \mathcal{P}} \left[h_{j}(\pi(\gamma)) - \beta_j \right]\right\} , 
\end{align*}
where the minimum is taken over all the possible functions mapping inputs to actions.
The Slater's parameter quantifies the strict feasibility margin of the problem. Intuitively, it represents how well constraints are satisfied by the ``most feasible'' solution.

We remark that the proposed definition of $\rho_{\text{stoc}}$ naturally reduces to the standard strict feasibility margin when applied to classical budget-constrained problems (\emph{e.g.}, \cite{balseiro2023best}). In such settings, evaluating our definition under the constant policy $\pi(\gamma) = \varnothing $ for all $\gamma \in \Gamma$, the expectation over the selected action trivially yields zero consumption. Consequently, for each resource $j \in [n]$, the expectation over the input distribution evaluates deterministically to $\mathbb{E}_{\gamma \sim \mathcal{P}}[0 - \beta_j] = -\beta_j$, which yields $\rho = \min_{j \in [n]} \beta_j$
Thus, the expectation over $\gamma$ in our formulation provides the generalization required for settings without a deterministically safe action.

In the stochastic setting, we measure the performance of the agent over the time horizon $T$ by means of the \emph{expected cumulative regret} (or simply \emph{regret} for short), which is defined as follows: 
\[R_T := \text{OPT}_\text{stoc} - \mathbb{E}_{\bm{\gamma} \sim \mathcal{P}^T}[\texttt{Rew}(T, \bm{\gamma})].\] 

\subsection{Adversarial Setting} \label{p_a}

In the adversarial setting, the dynamic-optimum is defined as $\text{OPT}_\text{adv} = \text{OPT}(\bm{\gamma})$ for a given fixed input sequence $\bm{\gamma} = (\gamma_1, \ldots, \gamma_T) \in \Gamma^T$.
Moreover, in this setting, the \emph{Slater's parameter} is:
\[\rho_{\text{adv}} = - \min_{(x_1,...,x_T) \in \mathcal{X}^T} \max_{t \in [T]} \max \left\{ \max_{i \in [m]} g_{t,i}(x_t), \max_{j \in [n]} h_{t,j}(x_t) - \beta_j \right\}.\] 
As shown by~\citet{balseiro2023best}, it is impossible to obtain both sublinear regret and sublinear cumulative constraint violation when constraints are selected adversarially. Consequently, we cannot rely on the standard notion of regret for the adversarial setting. To overcome this limitation, we study the cumulative regret with respect to a fraction $\alpha \in (0,1)$ of the optimum, called $\alpha\text{-}$regret: 
\[
\alpha\text{-}R_T := \alpha \cdot\text{OPT}_\text{adv} - \texttt{Rew}(T),
\]  
where $\alpha := \nicefrac{\rho_\text{adv}}{1 + \rho_\text{adv}}$ and we let $\texttt{Rew}(T)\coloneqq \texttt{Rew}(T,\bm{\gamma})$ since the dependence of (the deterministic) $\gamma$ is clear from the context.

We remark that, both in the stochastic setting introduced in Section~\ref{p_s} and the adversarial one discussed in Section~\ref{p_a}, the agent has no prior knowledge of the Slater's parameter.

\subsection{Lagrangian Formulation} 
For the sake of exposition, it is convenient to group the general long-term constraints and the budget constraints into a single unified vector. Let $M \coloneqq m + n$ be the total number of constraints. For each round $t \in [T]$ and action $x \in \mathcal{X}$, we define the unified constraint vector $\tilde{\bm{g}}_t(x) \in \mathbb{R}^M$ as follows: 
\[
\tilde{g}_{t,i}(x) =
\begin{cases}
	g_{t,i}(x) & \text{if } i \in [m] \\
	h_{t, i-m}(x) - \beta_{i-m}  & \text{if } m < i \leq M.
\end{cases}
\]
A fundamental tool for addressing constrained optimization problems is the Lagrangian formulation. Given the instantaneous reward function $f_t$, and the unified constraint function $\tilde{\bm{g}}_t$, at round $t \in [T]$, we define the instantaneous Lagrangian function as: 
\[\mathcal{L}_{f_t,\bm{\tilde{g}}_t}(x, \bm{\lambda}) := f_t(x) - \sum_{i=1}^{M} \lambda_i \tilde{g}_{t,i}(x) = f_t(x) - 
\langle \bm{\lambda}, \bm{\tilde{g}}_t(x) \rangle,\] 
where $\bm{\lambda} \in \mathbb{R}^M_{\geq 0}$ is the vector of dual multipliers.

\section{Algorithm} \label{sec:alg} 
In this section, we describe our algorithm, whose pseudo-code is provided in Algorithm~\ref{alg:alg1}. 

\begin{algorithm}[!htp]
	\caption{\textsc{Dual Gradient Descent}}
	\label{alg:alg1}
	\begin{algorithmic}[1]
		\Require $T$, $\eta$, $B_j$ for all $j \in [n]$. \label{algLine:0}
		\State $\bm{\lambda}_1\gets \bm{0}$
		\For{$t = 1, 2, \ldots , T$}
		\State Observe input tuple $\gamma_t = (f_t, \bm{g}_t, \bm{h}_t)$ \label{algLine:2}
		\State $\widehat{x}_t \gets
		\argmax_{x \in \mathcal{X}}\mathcal{L}_{f,\bm{\tilde{g}}}(x, \bm{\lambda}_t)$ \Comment{Compute candidate action}\label{algLine:3}
		\If{$\sum_{s = 1}^{t-1} h_{s, j}(x_s) \leq \beta_j T-1$ for every $j \in[n]$}\label{algLine:4}
		\State $x_t \gets
		\widehat{x}_t$ \label{algLine:5}
		\Else 
		\State $x_t \gets \varnothing$ \label{algLine:7}
		\EndIf
		\State $\bm{\lambda}_{t+1} \gets \Pi_{\mathbb{R}^M_{\geq 0}}(\bm{\lambda}_t + \eta\bm{\tilde{g}}_t(x_t))$ \Comment{Dual update}\label{algLine:10}
		\EndFor
	\end{algorithmic}
\end{algorithm}

Algorithm~\ref{alg:alg1} operates as follows. It receives as input the time horizon $T$, the learning rate for the dual update $\eta > 0$, and the budget $B_j$ for each resource $j \in [n]$.
At the beginning of each round $t \in [T]$, the input $\gamma_t = (f_t, \bm{g}_t, \bm{h}_t)$ is fully revealed (Line~\ref{algLine:2}). Then, the algorithm computes the candidate action $\widehat{x}_t$ by greedily maximizing the instantaneous Lagrangian function (Line~\ref{algLine:3}). Since budget constraints must be strictly satisfied, the algorithm checks, for each resource $j \in [n]$, whether the cumulative resource consumption up to the previous round exceeds the budget (Line~\ref{algLine:4}). If the budget constraints are satisfied for all resources, the algorithm plays the candidate action $\widehat{x}_t$ (Line~\ref{algLine:5}). Once the budget is exceeded for at least one resource, the algorithm is forced to select the void action $\varnothing$ (Line~\ref{algLine:7}).\footnote{Technically, we stop slightly before in order to guarantee that the budget constraints are exactly satisfied.} Finally, after the primal action is played, the algorithm performs an \emph{online gradient descent} (OGD) step to update and project the dual variable $\bm{\lambda}$ onto the non-negative orthant $\mathbb{R}^M_{\geq 0}$ for the next round (Line~\ref{algLine:10}). 

In the following, it will be useful to define the stopping time $\tau \leq T$, i.e., the last round in which $\sum_{t \in [\tau]} h_{t,j}(x_t)  \geq \beta_j T-1$ for each resource $j\in [n]$.  For any $t > \tau$, the agent play the void action $x_t = \varnothing$.

\section{Bounding the Lagrangian Multipliers} \label{sec:dual}

In this section, we present the primary technical contribution of our work. Specifically, we show that the Lagrange multipliers remain bounded during the learning dynamic. This analysis requires novel techniques with respect to the ones employed in budget-constrained problems, as general constraints are much more prone to conflict with one another, pushing the primal decision in opposing directions. For instance, while playing the void action is the trivial way to satisfy all budget constraints simultaneously, this property does not hold for general functions. We provide a detailed discussion of this distinction in \Cref{sec:ex} and establish the upper bounds for the multipliers in \Cref{dual_bound}.

\subsection{Difference Between Budget and General Constraints}\label{sec:ex}

Classical ORA problems, such as those studied in \citep{balseiro2023best}, focus exclusively on budget constraints. In such settings, bounding the dual variables follows directly from the definition of void action: playing $\varnothing$ yields zero resource consumption ($h_{t,j}(\varnothing) = 0$), guaranteeing a strictly positive and a priori known feasibility margin (the per-iteration budget $\beta_j$). In absence of other constraints, the overall feasibility margin of the problem simply coincides with the minimum budget parameter ($\rho = \min_j \beta_j > 0$). This margin allows previous works to analytically bound dual multipliers. Our setting, however, incorporates general constraints $\bm{g}_t \in [-1,1]^m$ alongside budget constraints. For these general constraints, the void action yielding $g_{t,i}(\varnothing) = 0$ provides zero slack. Consequently, the overall Slater's parameter no longer coincides with the minimum $\beta_j$, rather, it becomes a problem-dependent unknown quantity. 
To highlight the fundamental technical difference between standard budget-constrained problems (such as \citep{balseiro2023best}) and our setting with general constraints, we present a simplified single-round example with two constraints (\( m = 2 \)).
\begin{example} \label{ex}
	Let us consider an ORA problem instance.\footnote{ We remark that while this reasoning holds for the adversarial setting, it can be trivially extended to the stochastic setting.}
	Suppose that the dual multipliers at the beginning of some round $t \in [T]$ are large and equal $\bm{\lambda}_t = [\lambda_{t,1} \ \lambda_{t,2}] \simeq [ \nicefrac{2}{\rho} \ \  \nicefrac{2}{\rho}]$. We analyze how a standard dual-based algorithm that maximizes the Lagrangian behaves in two different settings: the first with only budget constraints, and the second with only general constraints.
	\begin{itemize}[leftmargin=0.5cm]
		\item In the classical setting with only budget constraints, the constraint function is $h_{t,i}(x) - \rho$, namely, it takes values in $[-\rho, 1 - \rho]$. This means that the maximum extent to which an action can ``satisfy'' the budget constraint is bounded by $-\rho$. The agent has access to a deterministic void action $\varnothing$, which yields zero rewards and consumes zero budget. Consider any other action in $\mathcal{X}$ that incurs a violation greater than $0$ in any of the constraints. For instance, action $x_A \in \mathcal{X}$ such that $f(x_A) = 1$ and $\bm{h}_t(x_A) = [\rho+\epsilon \  0]^\top$. We compare the Lagrangian objective obtained by playing the two actions $\varnothing$ and $x_A$ at round $t$: 
		\begin{align*}
			\mathcal{L}_{f_t, \bm{\tilde{g}}_t}&(\varnothing, \bm{\lambda}_t) = f_t(\varnothing) - \langle \bm{\lambda}_t, \bm{h}_t(\varnothing) - \rho \bm{1} \rangle = 0 + \langle \bm{\lambda}_t, \rho \bm{1} \rangle  = \lambda_{t,1} \rho + \lambda_{t,2} \rho = 4 >\\
			&\mathcal{L}_{f_t, \bm{\tilde{g}}_t}(x_A, \bm{\lambda}_t) = f_t(x_A) - \langle \bm{\lambda}_t, \bm{h}_t(x_A) - \rho \bm{1} \rangle = 1 - \lambda_{t,1} (\rho+\epsilon -\rho)  + \lambda_{t,2} \rho = 3 - \frac{2 \epsilon}{\rho}.
		\end{align*}
		Although $x_A$ yields the best per-round reward possible and satisfies one of the constraint as much as possible, the slight violation on the other constraint penalizes the Lagrangian objective with respect to satisfying both the constraints simultaneously. This holds for any action that violates one of the constraints.
		Since no other action could yield a better Lagrangian value, the algorithm is encouraged to play $\varnothing$, and therefore it stops violating.
		
		\item In the setting with general constraints, 
		constraints take values in $[-1, 1]$ and are not aligned. Indeed, extreme cross-constraint compensation in the Lagrangian may happen. The ``safest'' action in this setting is the strictly safe action leading to the computation of $\rho$, i.e. $x^\circ$ such that $\bm{g}_t(x^\circ) = [-\rho \ \  -\rho]^\top$ by definition. 
		Consider action $x_B \in \mathcal{X}$, yielding $f(x_B) = 1$ and $\bm{g}_t(x_B) = [\rho + \epsilon \ \ -1]^\top$. Similarly to $x_A$, $x_B$ provides the best reward possible and satisfies as much as possible one constraints while slightly violating the other.
		
		Comparing the Lagrangian value obtained by playing either of the two actions $x^\circ$ and $x_B$, we obtain the following inequality: 
		\begin{align*}
			\mathcal{L}_{f_t, \bm{\tilde{g}}_t}&(x^\circ, \bm{\lambda}_t) = f_t(x^\circ) - \langle \bm{\lambda}_t, \bm{g}_t(x^\circ) \rangle \leq 1 + \lambda_{t,1} \rho + \lambda_{t,2} \rho = 5
			< \\
			&\mathcal{L}_{f_t, \bm{\tilde{g}}_t}(x_B, \bm{\lambda}_t) = f_t(x_B) - \langle \bm{\lambda}_t, \bm{g}_t(x_B)  \rangle = 1 - \lambda_{t,1} \cdot (\rho + \epsilon)  + \lambda_{t,2} \cdot 1 \overset{(3\rho + \epsilon = \frac{1}{2})}{=} 5 + \frac{1}{\rho}.
		\end{align*}
		When $3\rho + \epsilon < 1$, since the satisfaction of the second constraint is much larger than $\rho$, its contribution to the Lagrangian function 
		is large and strictly positive although the first constraint is violated. The strictly safe action, which only provides margin $\rho$, yields a worse contribution. Therefore, the algorithm will prefer to play $x_B$ (\emph{i.e.}, keep violating) with respect to the strictly safe action, even when $x^\circ$ leads to the maximum possible reward $f_t(x^\circ)=1$.
	\end{itemize}
\end{example}
In Figure~\ref{fig:lambda-comparison-clean}, we provide a graphical representation of these aspects.

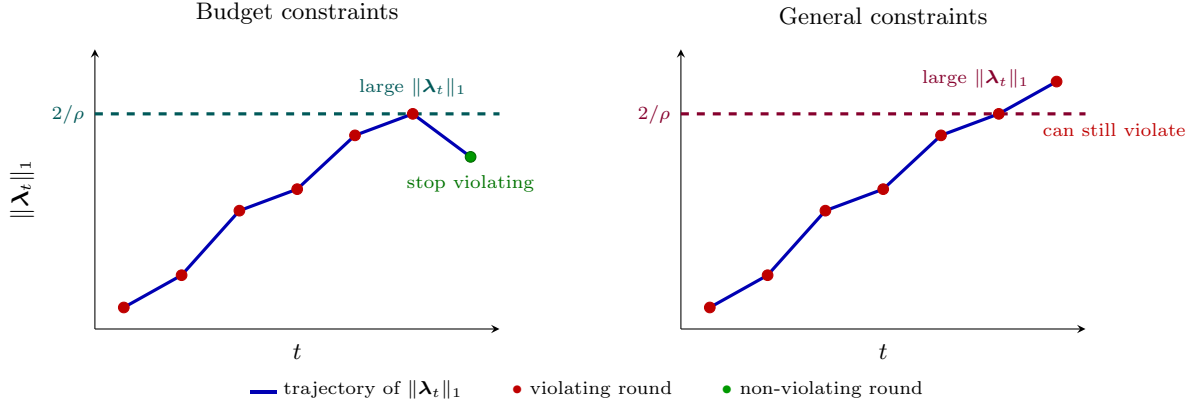
\begin{figure}[t]
	\centering
	\begin{tikzpicture}
		\begin{groupplot}[
			group style={group size=2 by 1, horizontal sep=2.4cm},
			width=0.42\textwidth,
			height=0.32\textwidth,
			xmin=0.5, xmax=7.5,
			ymin=0, ymax=26,
			xtick=\empty,
			ytick=\empty,
			xlabel={$t$},
			axis lines=left,
			label style={font=\small},
			title style={font=\small},
			clip=false,
			grid=none,
			]
			
			\nextgroupplot[
			title={Budget constraints},
			ylabel={$\|\bm{\lambda}_t\|_1$},
			ylabel style={
				rotate=0,
				at={(axis description cs:-0.18,0.5)},
				anchor=center
			},
			]
			
			\addplot[
			very thick,
			dashed,
			teal!70!black
			] coordinates {(0.5,20) (7.5,20)};
			
			\node[
			font=\scriptsize,
			text=teal!70!black,
			anchor=east
			] at (axis cs:0.48,20) {$2/\rho$};
			
			\addplot[
			very thick,
			blue!70!black,
			mark=none
			] coordinates {
				(1,2) (2,5) (3,11) (4,13) (5,18) (6,20) (7,16)
			};
			
			\addplot[
			only marks,
			mark=*,
			mark size=2.0pt,
			fill=red!75!black,
			draw=red!75!black
			] coordinates {
				(1,2) (2,5) (3,11) (4,13) (5,18) (6,20)
			};
			
			\addplot[
			only marks,
			mark=*,
			mark size=2.0pt,
			fill=green!60!black,
			draw=green!40!black
			] coordinates {
				(7,16)
			};
			
			\node[
			font=\scriptsize,
			text=teal!70!black,
			align=center
			] at (axis cs:6.0,22.5) {large $\|\bm{\lambda}_t\|_1$};
			
			\node[
			font=\scriptsize,
			text=green!45!black,
			align=center
			] at (axis cs:7.0,13.5) {stop violating};
			
			\nextgroupplot[
			title={General constraints},
			ylabel={}
			]
			
			\addplot[
			very thick,
			dashed,
			purple!70!black
			] coordinates {(0.5,20) (7.5,20)};
			
			\node[
			font=\scriptsize,
			text=purple!70!black,
			anchor=east
			] at (axis cs:0.48,20) {$2/\rho$};
			
			\addplot[
			very thick,
			blue!70!black,
			mark=none
			] coordinates {
				(1,2) (2,5) (3,11) (4,13) (5,18) (6,20) (7,23)
			};
			
			\addplot[
			only marks,
			mark=*,
			mark size=2.0pt,
			fill=red!75!black,
			draw=red!75!black
			] coordinates {
				(1,2) (2,5) (3,11) (4,13) (5,18) (6,20) (7,23)
			};
			
			\node[
			font=\scriptsize,
			text=purple!70!black,
			align=center
			] at (axis cs:5.6,23.3) {large $\|\bm{\lambda}_t\|_1$};
			
			\node[
			font=\scriptsize,
			text=red!75!black,
			align=center
			] at (axis cs:8.0,18.5) {can still violate};
			
		\end{groupplot}
		
		\node[
		font=\scriptsize,
		align=center
		] at ($(group c1r1.south)!0.5!(group c2r1.south)+(0,-0.8)$)
		{
			\textcolor{blue!70!black}{\rule{0.35cm}{1.2pt}} trajectory of $\|\bm{\lambda}_t\|_1$
			\qquad
			\textcolor{red!75!black}{$\bullet$} violating round
			\qquad
			\textcolor{green!60!black}{$\bullet$} non-violating round
		};
		
	\end{tikzpicture}
	\caption{
		Illustration of the behavior of standard dual algorithms that \emph{do not use a weakly adaptive regret minimizers}.
		Left: in the budget-constrained setting, once $\|\bm{\lambda}_t\|_1$ becomes sufficiently large, the algorithm stops violating, and the dual variable starts decreasing.
		Right: under general constraints, violations may persist even when $\|\bm{\lambda}_t\|_1$ is already large.
	}
	\label{fig:lambda-comparison-clean}
\end{figure}
\subsection{A Linear Upper Bound} \label{dual_bound} 
In standard Lagrangian frameworks, the magnitude of dual variables is typically controlled through an explicit projection step onto a bounded domain. However, this operation requires the knowledge of the Slater's parameter $\rho$. To bypass this restrictive assumption, we leverage the property of \emph{weak adaptivity}~\citep{HazanOCO}, which we formally defined as follows.
\begin{definition}[Weakly-adaptive regret minimizer] \label{def} 
	An online learning algorithm is said to be a \emph{weakly-adaptive regret minimizer} if, for any interval $I = [t_1, t_2] \subseteq [T] $, its regret $R_I$ accumulated over the interval is sublinear in $T$. Formally: 
	\[
	\max_{I=[t_1,t_2] \subseteq [T]} R_I(\mathcal{X}) = \max_{I= [t_1,t_2] \subseteq [T]} \left\{ \max_{x \in \mathcal{X}} \sum_{t \in I} f_t(x) - \sum_{t \in I} f_t(x_t) \right\} \leq o(T).
	\]
\end{definition}
Intuitively, a weakly-adaptive regret minimizer guarantees sublinear regret not just over the full time horizon $T$, but over any arbitrary sub-interval $[t_1, t_2] \subseteq [T]$. Online gradient descent, which we employ to update the Lagrangian multipliers, naturally satisfies this property for linear losses. 
In this section, we use weak adaptivity to provide a theoretical bound on the $\ell_1$-norm of the Lagrangian multipliers $\bm{\lambda}_t$, which allows us to subsequently bound the cumulative constraint violation. 
\begin{restatable}{lemma}{selfb} \label{selfb}
	Let $\delta\in(0,1)$. Consider Algorithm~\ref{alg:alg1} with learning rate $\eta = \nicefrac{1}{60 M \sqrt{2T \ln\left( \nicefrac{T^2}{\delta} \right)}}$.
	With probability at least $1-\delta$, the $\ell_1$-norm of the dual variables is bounded for all $t \in [T]$ as:
	\[ \|\bm{\lambda}_t\|_1 \leq \frac{14M}{\rho}.\] 
\end{restatable} 
To prove Lemma~\ref{selfb}, we proceed by contradiction. Supposing that the $\ell_1$-norm of the dual variables $\| \bm{\lambda}_t \|_1$ exceeds a threshold of $\mathcal{O}(M/\rho)$ at some round $t_2$, we analyze the resulting primal-dual dynamics. Intuitively, the dual variables act as penalty weights for constraint violations. When these multipliers grow disproportionately large, the penalty term $\langle \bm{\lambda}_t, \bm{\tilde{g}}_t(x) \rangle$ in the instantaneous Lagrangian objective completely dominates the reward, which is bounded in $[0,1]$. Consequently, to maximize the Lagrangian, the greedy primal step is forced to select an action that yields a penalty at least as favorable as the strictly feasible baseline. Since the baseline guarantees a strictly negative value, the primal response systematically pulls the multipliers back down.
We formalize this intuition by comparing the Lagrangian value of the selected action with that of the strictly feasible baseline, obtaining a lower bound on the cumulative Lagrangian $\sum_{t \in [t_1, t_2]} \mathcal{L}_{f_t, \bm{\tilde{g}}_t}(x_t, \bm{\lambda}_t)$. On the other hand, we derive an upper bound on the same quantity by exploiting the weak adaptivity of online gradient descent. Specifically, rather than evaluating the dual regret globally over the full horizon $T$, we evaluate it locally over the sub-interval $[t_1, t_2]$. 
This step highlights why our framework specifically requires OGD rather than a general online mirror descent (OMD) method. Algorithms guaranteeing only \emph{global} no-regret over $[T]$ (such as standard OMD with entropy regularizer), can accumulate negative regret during early rounds and exploit it to ``absorb'' later large penalties. Our weak adaptivity property prevents the algorithm from relying on past negative regret. This represents a fundamental difference with respect to \citep{balseiro2023best}. 

We show that for the dual variables to reach the assumed threshold, they would need to be updated at a rate strictly incompatible with our hypothesis. This means that the greedy primal is guaranteed to react and halt constraint violation long before the multiplier can ever reach that threshold. The extended proof can be found in Appendix~\ref{app:dual}.

We note that for $t > \tau$, the algorithm can only play the void action $\varnothing$, for which it holds $\tilde{g}_{t,i}(\varnothing) \leq 0$ for all $i \in [M]$. This means that the dual update can only decrease the multipliers (or keep them unchanged) and thus the analysis can focus on the interval $[1,\tau]$.

\section{Theoretical Results} \label{sec:theo} 
In this section, we prove the theoretical guarantees attained by Algorithm~\ref{alg:alg1}. In particular, given the previous bound on the Lagrangian variables, we analyze the constraint violation, the cumulative regret in the stochastic setting, and the cumulative $\alpha\text{-}$regret in the adversarial setting.

\subsection{Constraint Violation} \label{sec:violation}  
Using Lemma~\ref{selfb}, we can now establish the worst-case bound on the constraint violation $V_T$.

\begin{restatable}{theorem}{vt} \label{vt}
	Let $\delta\in(0,1)$. Both in the stochastic and the adversarial setting, Algorithm~\ref{alg:alg1} with learning rate $\eta := \nicefrac{1}{60 M \sqrt{2T \ln\left( \nicefrac{T^2}{\delta} \right)}}$, attains with probability at least $1 - \delta$: 
	\[ V_T \leq 840 \frac{M^2}{\rho} \sqrt{2T \ln\left( \frac{T^2}{\delta} \right)},\] 
	where $\rho \coloneqq \rho_\textnormal{stoc}$ in the stochastic setting and $\rho \coloneqq \rho_\textnormal{adv}$ in the adversarial setting. Moreover, the budget constraints are deterministically satisfied, i.e., \[\sum_{t \in [T]} h_{t,j}(x_t) \leq \beta_j T \quad \forall j \in [n].\]
\end{restatable} 

Theorem~\ref{vt} states that our algorithm guarantees, in both the stochastic and the adversarial setting, a constraint violation of the order of $\widetilde{\mathcal{O}}(\sqrt{T})$ without knowing the Slater's parameter. Intuitively,
by construction of the algorithm, budget constraints are never violated, since once the cumulative consumption of any resource $j$ reaches $\beta_j T - 1$, the algorithm starts selecting the void action $\varnothing$. Hence, it is sufficient to bound the violation of the general constraints $g_{t,i}$ for all $i \in [m]$. The bound is obtained by applying Lemma~\ref{selfb} and using the OGD update definition $\lambda_{t+1, i} = \max\{0, \lambda_{t,i} + \eta g_{t,i}(x_t) \}$. Telescoping the update rule over all rounds $t$, we obtain that the cumulative violation can be upper-bounded by a factor proportional to the final value of the dual multiplier. Intuitively, every violation increases the corresponding multiplier, and therefore, large cumulative violations imply large dual variables. Since the void action satisfies $g_{t,i}(\varnothing) = 0$ for all $t > \tau$, the bound is studied on the violation accumulated up to $\tau$. The full proof can be found in Appendix~\ref{app:v}.

\subsection{Stochastic Setting} \label{sec:stoch}   
In this section, we evaluate the performance of Algorithm~\ref{alg:alg1} in the stochastic regime, where inputs are drawn from a fixed but unknown distribution.

\begin{restatable}{theorem}{stre} \label{st_re}
	Let $\delta\in(0,1)$. Consider Algorithm~\ref{alg:alg1} with learning rate $\eta = \nicefrac{1}{60 M \sqrt{2T \ln\left( \nicefrac{T^2}{\delta} \right)}}$. With probability at least $1 - \delta$, in the stochastic setting it holds: 
	\[    R_T \leq \frac{1}{\beta_{\text{min}}} + \frac{60 M \sqrt{2 T \ln\left( \frac{T^2}{\delta} \right)}}{2 \beta_{\text{min}}^2} + \frac{\sqrt{T}}{120 \sqrt{2 \ln\left( \frac{T^2}{\delta} \right)}}.\]  
\end{restatable}  

Theorem~\ref{st_re} shows that Algorithm~\ref{alg:alg1} achieves the desired $\widetilde{\mathcal{O}}(\sqrt{T})$ regret bound for the stochastic setting. The analysis proceeds in two main steps. First, we establish an upper bound on the baseline $\text{OPT}_{\text{stoc}}$ using weak duality and the definition of expected Lagrangian $\bar{\mathcal{L}}_{f,\bm{\tilde{g}}}(\bm{\lambda} | \mathcal{P}) := \mathbb{E}_{(f, \bm{\tilde{g}}) \sim \mathcal{P}} [\mathcal{L}_{f,\bm{\tilde{g}}}(x_t, \bm{\lambda})]$. Then, we analyze the algorithm's performance up to the random stopping time $\tau$. We lower-bound the instantaneous algorithm reward $f_t(x_t)$ and consider its expected value conditional on the past $\Gamma_{t-1} = \{\gamma_1,..., \gamma_{t-1} \}$. We then relate the cumulative reward to the instantaneous Lagrangian by applying a martingale argument on both the reward and the penalty term. Leveraging the regret guarantees of the OGD update, we bound the penalty term $-\sum_{t \in [\tau]} \langle \bm{\lambda}_t, \bm{\tilde{g}}_t(x_t) \rangle$ by exploiting the fact the, when the budget of some resource $k$ is exhausted, the cumulative consumption provides a lower bound on the dual penalty. The extended analysis can be found in Appendix~\ref{app:st}.

\subsection{Adversarial Setting} \label{sec:adv} 
Finally, in this section, we analyze the performance of the algorithm in the adversarial setting, when inputs may arbitrarily change at each round. 

\begin{restatable}{theorem}{advreg} \label{adv_reg}
	Let $\delta\in(0,1)$. Consider Algorithm~\ref{alg:alg1} with learning rate $\eta = \nicefrac{1}{60 M \sqrt{2T \ln\left( \nicefrac{T^2}{\delta} \right)}}$. With probability at least $1 - \delta$, in the adversarial setting it holds: 
	\[\alpha\text{-}R_T \leq  \frac{1}{\beta_{\text{min}}} + \frac{60 M \sqrt{2 T \ln\left( \frac{T^2}{\delta} \right)}}{2 \beta_{\text{min}}^2} + \frac{\sqrt{T}}{120 \sqrt{2 \ln\left( \frac{T^2}{\delta} \right)}}.\] 
\end{restatable}
Theorem~\ref{adv_reg} proves that our algorithm attains a bound of $\widetilde{\mathcal{O}}(\sqrt{T})$ for the $\alpha\text{-}$regret in the adversarial setting. The result follows from the regret guarantees of the unconstrained OGD update. We first split the $\alpha$-regret into the contribution up to the stopping time $\tau$ and the contribution after $\tau$. After $\tau$, the algorithm is forced to play $\varnothing$, and therefore the regret is trivially bounded by $T - \tau$. 
For the rounds up to $\tau$, we exploit the fact that the primal decision maximizes the instantaneous Lagrangian. To define the $\alpha$-regret, we evaluate the primal performance against a specifically tailored reference sequence $(\xi_t^\diamond)_{t=1}^T$. Specifically, for each round $t \in [T]$, we define the strategy $\xi_t^\diamond$ as the one which plays $x^*_t$ with probability $\alpha$ and $x^\circ_t$ with probability $1-\alpha$, where $\alpha = \nicefrac{\rho_\textnormal{adv}}{1 + \rho_\textnormal{adv}}$, $(x^*_t)_{t=1}^T$ is the optimal sequence of actions and $(x^\circ_t)_{t=1}^T$ is the sequence of strictly feasible actions that leads to the computation of $\rho_\textnormal{adv}$. By definition of $\alpha$, the sequence $(\xi^\diamond_t)_{t=1}^T$ is feasible, in expectation. This allows us to reduce the analysis to the cumulative penalty term induced by the algorithm $-\sum_{t \in [\tau]} \langle \bm{\lambda}_t, \bm{\tilde{g}}_t(x_t) \rangle$. As in Theorem~\ref{st_re}, this dual penalty is bounded by exploiting the regret guarantees of the OGD update. The detailed proof can be found in Appendix~\ref{app:adv}.

\section{Discussion and Open Problems}

In this work, we study a generalized online resource allocation setting that incorporates both budget constraints and general constraints. We develop an algorithm that achieves best-of-both-world guarantees in this broader framework. Specifically, against the dynamic benchmark, it attains $\widetilde{\mathcal O}(\sqrt{T})$ regret in the stochastic regime and $\alpha$-regret of order $\widetilde{\mathcal O}(\sqrt{T})$ in the adversarial regime, while ensuring strict satisfaction of the budget constraints and sublinear cumulative violation for the general ones. 

An interesting direction for future work is to understand whether the results provided by this work in the adversarial regime can be improved in non-stationary settings, where the environment may change over time but the amount of non-stationarity is assumed to be bounded.

\bibliographystyle{plainnat}
\bibliography{example_paper}

\newpage
\appendix

\section{Preliminary Results} 
\begin{lemma} \label{ogd}
	Let the dual algorithm be online gradient descent on the positive orthant $\mathbb{R}^{M}_+$ with learning rate $\eta > 0$. Then for any comparator $\bm{\mu} \in \mathbb{R}^{M}_+$, for any interval $[t_1, t_2] \subseteq [T]$ it holds:
	\[\sum_{t \in [t_1, t_2]} \langle \bm{\lambda}_t, \bm{\tilde{g}}_t \rangle \geq \sum_{t \in [t_1, t_2]} \langle \bm{\mu}, \bm{\tilde{g}}_t \rangle - \frac{1}{2\eta} \|\bm{\lambda}_{t_1} - \bm{\mu} \|^2_2 - \frac{\eta}{2} TM. \] 
\end{lemma}

\begin{proof}
	By the properties of Euclidean projection onto a convex set:
	\begin{align*}
		\|\bm{\lambda}_{t+1} - \bm{\mu}\|_2^2 &= \|\Pi_{\Lambda} (\bm{\lambda}_t + \eta \bm{\tilde{g}}_t) - \bm{\mu} \|_2^2 \\ 
		&\leq \|\bm{\lambda}_t + \eta \bm{\tilde{g}}_t - \bm{\mu}\|_2^2 \\
		&= \| \bm{\lambda}_t - \bm{\mu} \|_2^2 + \eta^2 \|\bm{\tilde{g}}_t\|_2^2 + 2\eta \langle \bm{\lambda}_t - \bm{\mu}, \bm{\tilde{g}}_t \rangle. 
	\end{align*} 
	Rearranging the terms: 
	\begin{align*}
		2\eta \langle \bm{\mu} - \bm{\lambda}_t, \bm{\tilde{g}}_t \rangle &\leq \| \bm{\lambda}_t - \bm{\mu} \|_2^2 - \|\bm{\lambda}_{t+1} - \bm{\mu}\|_2^2 + \eta^2 \|\bm{\tilde{g}}_t\|_2^2 \\ 
		\langle \bm{\mu} - \bm{\lambda}_t, \bm{\tilde{g}}_t \rangle &\leq \frac{1}{2\eta} \| \bm{\lambda}_t - \bm{\mu} \|_2^2 - \frac{1}{2\eta}\|\bm{\lambda}_{t+1} - \bm{\mu}\|_2^2 + \frac{\eta}{2} \|\bm{\tilde{g}}_t\|_2^2.
	\end{align*}
	Summing over $t$, we obtain: 
	\begin{align*}
		\sum_{t \in [t_1, t_2]}\langle \bm{\mu} - \bm{\lambda}_t, \bm{\tilde{g}}_t \rangle &\leq \frac{1}{2\eta} \| \bm{\lambda}_{t_1} - \bm{\mu} \|_2^2 - \frac{1}{2\eta}\|\bm{\lambda}_{t_2 + 1} - \bm{\mu}\|_2^2 + \frac{\eta}{2} \sum_{t \in [t_1, t_2]} \|\bm{\tilde{g}}_t\|_2^2 \\ 
		&\leq \frac{1}{2\eta} \| \bm{\lambda}_{t_1} - \bm{\mu} \|_2^2 - \frac{1}{2\eta}\|\bm{\lambda}_{t_2 + 1} - \bm{\mu}\|_2^2 + \frac{\eta}{2} T M \\ 
		&\leq \frac{1}{2\eta} \| \bm{\lambda}_{t_1} - \bm{\mu} \|_2^2 + \frac{\eta}{2} T M .
	\end{align*}
	Rearranging the terms gives the stated bound.
\end{proof}

\begin{lemma} \label{r_d}
	Let $\{ \bm{\lambda}_t \}_{t=1}^\tau$ be the sequence of dual multipliers generated by the OGD update rule (Lemma~\ref{ogd}) with learning rate $\eta > 0$ and $\bm{\lambda}_1 = \bm{0}$. Let $\tau \in [T]$ be the first round such that there exists a resource triggering the stopping condition, or
	$\tau = T$ if this condition is never met. Then, the dual penalty is bounded as: 
		\[
		- \sum_{t \in [\tau]} \langle \bm{\lambda}_t, \bm{\tilde{g}}_t(x_t) \rangle \leq -(T - \tau) + \frac{1}{\beta_{\min}} + \frac{1}{2\eta \beta_{\min}^2} +  \frac{\eta}{2} TM.
		\]
\end{lemma} 
\begin{proof}
	We analyze two cases by invoking Lemma~\ref{ogd} on the interval $[ \tau]$ with different comparators $\mu \in \mathbb{R}_+^M$:
	\begin{itemize}
		\item ($\tau = T$) We chose the comparator $\bm{\mu} = \bm{0}$. Since $\bm{\lambda}_1 = \bm{0}$, Lemma~\ref{ogd} directly yields: \begin{align*}
			- \sum_{t \in [\tau]} \langle \bm{\lambda}_t, \bm{\tilde{g}}_t(x_t) \rangle &\leq -\sum_{t \in [\tau]} \langle \bm{\mu}, \bm{\tilde{g}}_t(x_t) \rangle + \frac{1}{2\eta} \|\bm{\lambda}_1 - \bm{\mu} \|^2_2 + \frac{\eta}{2} TM \\ 
			&\leq \frac{\eta}{2} TM \\
		\end{align*} 
		\item ($\tau < T$) In this case, there exist a resource $k \in [n]$ that is exhausted at round $\tau$. By definition of the unified constraint vector, we have: 
		\begin{align}\sum_{t \in [\tau]} \tilde{g}_{t,m+k}(x_t) &= \sum_{t \in [\tau]} (h_{t,k}(x_t) - \beta_k) \nonumber\\
			&= \sum_{t \in [\tau]} h_{t,k}(x_t) - \tau\beta_k  \nonumber\\
			&\geq  (\beta_k T - 1) - \tau\beta_k \nonumber\\
			&= \beta_k(T - \tau) - 1 .\label{ct}
		\end{align}
		We invoke Lemma~\ref{ogd} choosing the comparator $\bm{\mu}$ such that $\mu_{m+k} = \frac{1}{\beta_k}$ and $\mu_{j} = 0$ for each $j \neq m + k$. Thus, we have: 
		\begin{align} -\sum_{t \in [\tau]}\langle \bm{\lambda}_t, \bm{\tilde{g}}_t(x_t) \rangle &\leq - \sum_{t \in [\tau]} \langle \bm{\mu}, \bm{\tilde{g}}_t(x_t) \rangle + \frac{1}{2\eta} \|\bm{\lambda}_1 - \bm{\mu} \|^2_2 + \frac{\eta}{2} TM \nonumber \\
			&= -\frac{1}{\beta_k} \sum_{t \in [\tau]}  \tilde{g}_{t, m+k} + \frac{1}{2\eta \beta_k^2} +  \frac{\eta}{2} TM \nonumber \\
			&\leq -(T - \tau) + \frac{1}{\beta_k} + \frac{1}{2\eta \beta_k^2} +  \frac{\eta}{2} TM,  \nonumber
		\end{align}
		where the last step follows from Inequality~\eqref{ct}. Therefore, in all cases it holds: 
		\[
		- \sum_{t \in [\tau]} \langle \bm{\lambda}_t, \bm{\tilde{g}}_t(x_t) \rangle \leq -(T - \tau) + \frac{1}{\beta_k} + \frac{1}{2\eta \beta_k^2} +  \frac{\eta}{2} TM.
		\]
		Defining $\beta_{\min} := \min_{j \in [n]} \beta_j$, it holds $\beta_k \geq \beta_{\min}$. Thus, we can write:  
		\[
		- \sum_{t \in [\tau]} \langle \bm{\lambda}_t, \bm{\tilde{g}}_t(x_t) \rangle \leq -(T - \tau) + \frac{1}{\beta_{\min}} + \frac{1}{2\eta \beta_{\min}^2} +  \frac{\eta}{2} TM.
		\]
	\end{itemize} 
	This concludes the proof.
\end{proof}
\section{Omitted Proofs and Lemmas of Section~\ref{dual_bound}} \label{app:dual}
\begin{lemma} \label{aux1}
	Let $\bm{\lambda}_t \in \mathbb{R}_{\geq 0}^m$ be generated by \textit{OGD} with learning rate $\eta$ and utilities $\bm{\lambda} \to \langle \bm{\lambda}, \bm{\tilde{g}}_{t} \rangle$ where $\|\bm{\tilde{g}}_t\|_{\infty} \leq 1$ for all $t \in [T]$. Then: 
	\[ \big| \| \bm{\lambda}_{t+1} \|_1 - \| \bm{\lambda}_{t} \|_1 \big| \leq \eta \cdot M .\] 
\end{lemma}
\begin{proof}
	By definition, the OGD update can be written as: 
	\[\lambda_{t+1,i} := \max \left\{0, \lambda_{t,i} + \eta \tilde{g}_{t,i}\right\} \quad \forall i \in [M]. \] 
	For all $i \in [M]$, $\lambda_{t+1,i} = \lambda_{t,i} + \eta \tilde{g}_{t,i} \leq \lambda_{t,i} + \eta$ when $\tilde{g}_{t,i} \geq 0$, while $\lambda_{t+1,i} \geq \lambda_{t,i} + \eta \tilde{g}_{t,i} \geq \lambda_{t,i} - \eta$ when $\tilde{g}_{t,i} < 0$. 
	Therefore, in all cases it holds: 
	\[| \lambda_{t+1,i} - \lambda_{t,i} | \leq \eta \quad \forall i \in [M].\] 
	Taking the sum over all components we obtain: 
	\[\| \bm{\lambda}_{t+1} - \bm{\lambda}_{t} \|_1 \leq \eta \cdot M.\] 
	By the triangular inequality, we get: 
	\[ \big| \| \bm{\lambda}_{t+1} \|_1 - \| \bm{\lambda}_{t} \|_1 \big| \leq \eta \cdot M. \] 
	This concludes the proof.
\end{proof}

\begin{lemma} \label{aux2}
	In the stochastic setting, for any $\pi(\cdot) : \Gamma \to  \mathcal{X}$ and $\delta \in (0,1)$, the following holds with probability at least $1 - \delta$: 
	\[\sum_{t \in [t_1, t_2]} \langle \bm{\lambda}_t, \bm{\tilde{g}}_t(\pi(\gamma_t)) \rangle\leq \sum_{t \in [t_1, t_2]} \mathbb{E}_{\gamma_t =(f_t, \bm{\tilde{g}}_t) \sim \mathcal{P}} \left[ \langle \bm{\lambda}_t,  \bm{\tilde{g}}_t(\pi(\gamma_t) \rangle \right]  + 2L \sqrt{2 T \ln\left(\frac{T^2}{\delta}\right)},\] 
	where $L = \max_{t \in [T]} \|\bm{\lambda}_t\|_1$ and $t_1, t_2 \in [T] : t_1 < t_2$.
\end{lemma}
\begin{proof} 
	Let $\mathcal{F}_{t-1}$ be the filtration representing the history up to the end of round $t-1$. 
	Since: 
	\[ \mathbb{E} \left[ \langle \bm{\lambda}_t, \bm{\tilde{g}}_t(\pi(\gamma_t)) \rangle \;\middle|\; \mathcal{F}_{t-1} \right] = \mathbb{E}_{\gamma_t =(f_t, \bm{\tilde{g}}_t) \sim \mathcal{P}} \left[ \langle \bm{\lambda}_t,  \bm{\tilde{g}}_t(\pi(\gamma_t)) \rangle  \right], \]
	it follows that the quantity $\langle \bm{\lambda}_t, \bm{\tilde{g}}_t(\pi(\gamma_t)) \rangle - \mathbb{E}_{\gamma_t =(f_t, \bm{\tilde{g}}_t) \sim \mathcal{P}} \left[ \langle \bm{\lambda}_t,  \bm{\tilde{g}}_t(\pi(\gamma_t)) \rangle  \right]$ is a martingale difference. 
	By Azuma-Hoeffding inequality, given an interval $[t_1, t_2] \subseteq [T]$, with probability at least $1 - \delta'$ it holds: 
	\begin{align*}
		\sum_{t \in [t_1, t_2]} \langle \bm{\lambda}_t, \bm{\tilde{g}}_t(\pi(\gamma_t)) \rangle &- \sum_{t \in [t_1, t_2]} \mathbb{E}_{\gamma_t =(f_t, \bm{\tilde{g}}_t) \sim \mathcal{P}} \left[ \langle \bm{\lambda}_t,  \bm{\tilde{g}}_t(\pi(\gamma_t)) \rangle  \right]\\  
		&\leq 2L \sqrt{2(t_2 - t_1) \ln\left(\frac{1}{\delta'}\right)} \leq 2L \sqrt{2 T \ln\left(\frac{1}{\delta'}\right)}.\ 
	\end{align*}
	Taking $\delta = T^2 \delta'$, we obtain that with probability at least $1 - \delta$, by union bound over all possible intervals it holds: 
	\[\sum_{t \in [t_1, t_2]} \langle \bm{\lambda}_t, \bm{\tilde{g}}_t(\pi(\gamma_t)) \rangle - \sum_{t \in [t_1, t_2]} \mathbb{E}_{\gamma_t =(f_t, \bm{\tilde{g}}_t) \sim \mathcal{P}} \left[ \langle \bm{\lambda}_t,  \bm{\tilde{g}}_t(\pi(\gamma_t)) \rangle  \right] \leq 2L \sqrt{2 T \ln\left(\frac{T^2}{\delta}\right)}.\] 
	This concludes the proof.
\end{proof}

\begin{lemma} \label{b4}
	Given an interval $[t_1, t_2] \subseteq [T]$ such that $\|\bm{\lambda}_t\|_1 \in [\nicefrac{c_1}{\rho}, \nicefrac{c_2}{\rho}]$ on $[t_1, t_2]$, $\|\bm{\lambda}_{t_2}\|_1 \geq \nicefrac{c_2}{\rho}$ and $\|\bm{\lambda}_{t_1}\|_1 \leq \nicefrac{c_1}{\rho}$. Given a learning rate $\eta > 0$, it holds: 
	\[\sum_{t \in [t_1,t_2]} \langle \bm{\lambda}_t, \bm{\tilde{g}}_t(x_t) \rangle \geq \frac{M}{2\rho^2 \eta}.\]
\end{lemma}
\begin{proof}
	We analyze the cumulative penalty term during the interval $[t_1, t_2]$. For each constraint index $i \in [M]$, let $\bar{t}_i \in [t_1, t_2]$ be the last round such that the dual multiplier is zero. If $\lambda_{t,i} > 0$ for the entire interval, we simply set $\bar{t}_i = t_1$. Formally:
	\[
	\bar{t}_i = \max\left\{ t_1, \max_{\tau \in [t_2] : \lambda_{\tau, i}=0} \ \tau \right\}.
	\]
	We can now split the analysis into two separate sub-intervals: $[t_1, \bar{t}_i]$ and $[\bar{t}_i, t_2]$.
	\begin{itemize}
		\item $(t \in [t_1, \bar{t}_i])$: By definition, it can be either that $\lambda_{\bar{t}_i}=0$ or $\bar{t}_i = t_1$. In the latter case, $[t_1, \bar{t}_i] = \varnothing$ and the dual regret is trivially zero. In the former case, by Lemma~\ref{ogd} with comparator $\bm{\mu} = \bm{0}$ we obtain: 
		\begin{align}
			0 \leq \sum_{t \in [t_1,\bar{t}_i]} \lambda_{t,i} \tilde{g}_{t,i}(x_t) + \frac{\lambda^2_{t_1}}{2\eta} + \frac{1}{2}\eta T. \label{gd}
		\end{align}
		\item ($t \in [\bar{t}_i,t_2]$): Due to the definition of $\bar{t}_i$, gradient descent never projects the multiplier relative to constraint $i$, and we can write that
		\begin{align}
			\sum_{t \in [\bar{t}_i,t_2]} \tilde{g}_{t,i}(x_t) = \frac{\lambda_{t_2,i} - \lambda_{\bar{t}_i,i}}{\eta}. \label{gd2}
		\end{align}
		By Lemma~\ref{ogd}, choosing the comparator $\bm{\mu}$ such that $\mu_i = \nicefrac{1}{\rho}$ for all $i \in [M]$ we obtain: 
		\begin{align}
			\sum_{t \in [\bar{t}_i, t_2]} \lambda_{t,i} \tilde{g}_{t,i}(x_t) &\geq \sum_{t \in [\bar{t}_i,t_2]} \mu_i \tilde{g}_{t,i}(x_t) -\frac{(\mu_i - \lambda_{\bar{t}_i,i})^2}{2\eta}
			- \frac{1}{2}\eta T \nonumber \\
			&= \frac{\lambda_{t_2,i} - \lambda_{\bar{t}_i,i}}{\rho \eta}
			- \frac{(\mu_i - \lambda_{\bar{t}_i, i})^2}{2\eta}
			- \frac{1}{2}\eta T, \label{gd3}
		\end{align}
		where we used Equation~\eqref{gd2}.
	\end{itemize}
	We can now combine the results from both sub-intervals. Summing Inequalities~\eqref{gd}-\eqref{gd3} we obtain:
	\begin{align*}
		\sum_{t \in [t_1,t_2]} \lambda_{t,i} \tilde{g}_{t,i}(x_t) &\geq \frac{\lambda_{t_2,i} - \lambda_{\bar{t}_i,i}}{\rho\eta} - \frac{(\mu_i - \lambda_{\bar{t}_i,i})^2}{2\eta} - \frac{\lambda_{t_1}^2}{2\eta} - \eta T \\
		&\geq \frac{\lambda_{t_2,i} - \lambda_{\bar{t}_i,i}}{\rho\eta} - \frac{(\mu_i)^2 + \lambda_{\bar{t}_i,i}^2}{2\eta} - \frac{\lambda_{t_1}^2}{2\eta} - \eta T.
	\end{align*}
	Finally, by summing over all $i \in [M]$, and defining $\bm{\lambda}_{\bar{t}}$ as the vector that has $\lambda_{\bar{t}_i}$ as its $i$-th component, we get:
	\begin{align}
		\sum_{t \in [t_1,t_2]} \langle \bm{\lambda}_t, \bm{\tilde{g}}_t(x_t) \rangle
		&\geq \frac{\|\bm{\lambda}_{t_2}\|_1 - \|\bm{\lambda}_{\bar{t}}\|_1}{\rho\eta} - \frac{1}{2\eta} \left( \|\bm{\mu}\|_2^2 + \|\bm{\lambda}_{\bar{t}}\|_2^2 + \|\bm{\lambda}_{t_1}\|_2^2 \right) - \frac{M}{\eta} \label{b1} \\
		&\geq \frac{c_2}{\rho^2 \eta} - \frac{1}{\rho\eta}\|\bm{\lambda}_{t_1}\|_1 - \frac{1}{2\eta} \left( \|\bm{\mu}\|_2^2 + 2\|\bm{\lambda}_{t_1}\|_2^2 \right) - \frac{M}{\eta} \label{b2} \\ 
		&\geq \frac{c_2}{\rho^2 \eta} - \frac{1}{\rho\eta} \left( \frac{c_1}{\rho} + M\eta \right) - \frac{1}{2\eta} \left( \frac{M}{\rho^2} + 2\left( \frac{c_1}{\rho} + M\eta \right)^2 \right) - \frac{M}{\eta} \nonumber \\
		&\geq \frac{c_2}{\rho^2 \eta} - \frac{c_1 + 1}{\rho^2 \eta} - \frac{M}{2\rho^2 \eta} - \frac{2(c_1 + 1)^2}{2\rho^2 \eta} - \frac{M}{\eta} \label{b3} \\
		&\geq \frac{2c_2 - 6 - M - 18 - 2 M}{2\rho^2 \eta} \nonumber \\
		&\geq \frac{M}{2\rho^2 \eta}, \nonumber
	\end{align}
	where Inequality~\eqref{b1} holds by definition of $\eta \leq \nicefrac{1}{\sqrt{T}}$, Inequality~\eqref{b2} is due to the fact that $\|\bm{\lambda}\|_1 \geq c_2/\rho $ and $\|\bm{\lambda}_{\bar{t}}\|_1 \leq \|\bm{\lambda}_{t_1}\|_1$ and Inequality~\eqref{b3} follows from the condition that $\eta \leq \nicefrac{1}{\rho M}$.
	The last two inequalities hold due to the choice of parameters, that is, $c_1 = 2$ and $c_2 = 14M$. This concludes the proof.
\end{proof}

\selfb* 
\begin{proof}
	Let $c_1 := 2$, $c_2 := 14M$ and $\eta>0$ such that $\eta \leq \eta_{\text{OGD}}$. Moreover, let $\rho = \rho_{\text{stoc}}$ and $\rho = \rho_{\text{adv}}$, respectively in the stochastic and in the adversarial setting. By contradiction, let us suppose that at some point  $\|\bm{\lambda}_t\|_1 \geq \frac{c_2}{\rho}$. Let $t_2 \in [T]$ be the first round in which this condition happens. Moreover, let us define $t_1$ as the last round such that for all $t \in [t_1, t_2]$, it holds $\|\bm{\lambda}_t\|_1\in \left[ \frac{c_1}{\rho}, \frac{c_2}{\rho} \right]$. 
	
	From Lemma~\ref{aux1}, we have: 
	\[ \big| \| \bm{\lambda}_{t+1} \|_1 - \| \bm{\lambda}_{t} \|_1 \big| \leq M \eta.\]  
	Telescoping the sum for all $t \in [t_1, t_2]$ we obtain:  
	\[\| \bm{\lambda}_{t_2} \|_1 - \| \bm{\lambda}_{t_1} \|_1 \leq M \eta (t_2 - t_1).\]    
	Moreover, by definition of $t_1$ and $t_2$ we can write the following inequalities: 
	\[\frac{c_1}{\rho} \leq \| \bm{\lambda}_{t_1} \|_1 \leq \| \bm{\lambda}_{t_1-1} \|_1 + M \eta \leq \frac{c_1}{\rho} + M \eta, \]  
	\[\frac{c_2}{\rho} \leq \| \bm{\lambda}_{t_2} \|_1 \leq \| \bm{\lambda}_{t_2-1} \|_1 + M \eta \leq \frac{c_2}{\rho} + M \eta. \] 
	
	Putting everything together, we obtain: 
	\begin{align}
		\frac{c_2}{\rho} - \frac{c_1}{\rho} - M \eta \leq \| \bm{\lambda}_{t_2} \|_1 &- \| \bm{\lambda}_{t_1} \|_1 \leq M \eta (t_2 - t_1) \nonumber \\
		\frac{c_2 - c_1}{\rho} - M \eta &\leq  M \eta (t_2 - t_1) \nonumber \\
		t_2 - t_1 \geq \frac{c_2 - c_1}{\rho M \eta } &- 1 \geq \frac{c_2 - c_1}{2 \rho M \eta }. \label{dif} 
	\end{align}
	
	We now provide a lower bound for the quantity $ \sum_{t \in [t_1, t_2]} \mathcal{L}_{f_t,\bm{\tilde{g}}_t}(x_t, \bm{\lambda}_t)$, holding for both the stochastic and the adversarial setting. 
	
	Since the primal decision $x_t$ is such that $x_t = \argmax_{x \in \mathcal{X}} \mathcal{L}_{f_t,\bm{\tilde{g}}_t}(x, \bm{\lambda}_t)$, it point-wise dominates any other action. Let $\pi^\circ : \Gamma \to \mathcal{X}$ be the function leading to the computation of $\rho_\text{stoc}$, i.e. $\sum_{t \in [t_1, t_2]} \mathbb{E}_{\gamma_t =(f_t, \bm{\tilde{g}}_t) \sim \mathcal{P}} \left[ \bm{\tilde{g}}_t(\pi^\circ(\gamma_t))\right] \leq - \rho_{\text{stoc}}$ for all $i \in [M]$  
	Thus, in the stochastic setting, we can write: 
	\begin{align}
		\sum_{t \in [t_1, t_2]} &\mathcal{L}_{f_t,\bm{\tilde{g}}_t}(x_t, \bm{\lambda}_t) \\ 
		&\geq \sum_{t \in [t_1, t_2]}  \mathcal{L}_{f_t,\bm{\tilde{g}}_t}(\pi^\circ(\gamma_t), \bm{\lambda}_t) \nonumber \\ 
		&= \sum_{t \in [t_1, t_2]}  f_t(\pi^\circ(\gamma_t)) - \langle \bm{\lambda}_t, \bm{\tilde{g}}_t(\pi^\circ(\gamma_t)) \rangle \nonumber \\
		&\geq - \sum_{t \in [t_1, t_2]} \langle \bm{\lambda}_t, \bm{\tilde{g}}_t(\pi^\circ(\gamma_t)) \rangle \nonumber \\ 
		&\geq - \sum_{t \in [t_1, t_2]} \mathbb{E}_{\gamma_t =(f_t, \bm{\tilde{g}}_t) \sim \mathcal{P}} \left[  \langle \bm{\lambda}_t,  \bm{\tilde{g}}_t(\pi^\circ(\gamma_t)) \rangle \right] - 2L \sqrt{2 T \ln\left(\frac{T^2}{\delta}\right)} \label{ahb} \\ 
		&\geq \rho_{\text{stoc}} \sum_{t \in [t_1, t_2]} \| \bm{\lambda}_t \|_1 - 2L \sqrt{2 T \ln\left(\frac{T^2}{\delta}\right)} \label{r1} \\ 
		&\geq c_1 (t_2 - t_1) - 2\left( \frac{c_2}{\rho_{\text{stoc}}} + M\eta\right) \sqrt{2 T \ln\left(\frac{T^2}{\delta}\right)}, \nonumber 
	\end{align}    
	where Inequality~\eqref{ahb} holds by Lemma~\ref{aux2}, and Inequality~\eqref{r1} holds by definition of $\pi^\circ$. 
	
	In the adversarial setting, let $(x_t^\circ)_{t=1}^T$ be the sequence of strictly feasible actions that leads to the definition of $\rho_{\text{adv}}$, meaning $\tilde{g}_{t,i}(x^\circ_t) \leq -\rho_{\text{adv}}$ for all $i \in [M]$. It holds: 
	\begin{align}
		\sum_{t \in [t_1, t_2]} \mathcal{L}_{f_t,\bm{\tilde{g}}_t}(x_t, \bm{\lambda}_t) &\geq  \sum_{t \in [t_1, t_2]} \mathcal{L}_{f_t,\bm{\tilde{g}}_t}(x^\circ_t, \bm{\lambda}_t) \\ &=  \sum_{t \in [t_1, t_2]} f_t(x^\circ_t) - \langle \bm{\lambda}_t, \bm{\tilde{g}}_t(x^\circ_t) \rangle \nonumber \\
		&\geq -  \sum_{t \in [t_1, t_2]}  \langle \bm{\lambda}_t, \bm{\tilde{g}}_t(x^\circ_t) \rangle \nonumber \\ 
		&\geq \rho_{\text{adv}} \sum_{t \in [t_1, t_2]} \| \bm{\lambda}_t \|_1 \label{r2} \\ 
		&\geq c_1(t_2 - t_1), \nonumber
	\end{align} 
	where Inequality~\eqref{r2} holds by definition of $(x^\circ_t)_{t=1}^T$.
	
	Thus, for both settings it holds: 
	\begin{align} \sum_{t \in [t_1, t_2]} \mathcal{L}_{f_t,\bm{\tilde{g}}_t}(x_t, \bm{\lambda}_t) \geq c_1 (t_2 - t_1) - 2\left( \frac{c_2}{\rho} + M\eta\right) \sqrt{2 T \ln\left(\frac{T^2}{\delta}\right)}. \label{lb} \end{align}
	We employ Lemma~\ref{b4} to upper bound the following term $\sum_{t \in [t_1, t_2]} \mathcal{L}_{f_t,\bm{\tilde{g}}_t}(x_t, \bm{\lambda}_t)$:
	\begin{align}
		\sum_{t \in [t_1, t_2]} \mathcal{L}_{f_t,\bm{\tilde{g}}_t}(x_t, \bm{\lambda}_t) &= \sum_{t \in [t_1, t_2]} ( f_t(x_t) - \langle \bm{\lambda}_t, \bm{\tilde{g}}_{t}(x_t) \rangle ) \nonumber \\ 
		&\leq \sum_{t \in [t_1, t_2]} f_t(x_t)  -  \frac{M}{2 \rho^2 \eta} \nonumber \\
		&\leq (t_2 - t_1) -  \frac{M}{2 \rho^2 \eta}. \label{ub}
	\end{align}
	
	Putting together Inequalities~\eqref{lb}-\eqref{ub} we obtain: 
	\begin{gather}
		c_1 (t_2 - t_1) - 2\left( \frac{c_2}{\rho} + M\eta\right) \sqrt{2 T \ln\left(\frac{T^2}{\delta}\right)} \leq \sum_{t \in [t_1, t_2]} \mathcal{L}_{f_t,\bm{\tilde{g}}_t}(x_t, \bm{\lambda}_t) \leq  (t_2 - t_1) - \frac{M}{2 \rho^2 \eta} \nonumber \\
		c_1 (t_2 - t_1) - 2\left( \frac{c_2}{\rho} + M\eta\right) \sqrt{2 T \ln\left(\frac{T^2}{\delta}\right)} \leq (t_2 - t_1) - \frac{M}{2 \rho^2 \eta} \nonumber \\
		t_2 - t_1 \leq \frac{1}{c_1 - 1} \left[ 2\left( \frac{c_2}{\rho} + M\eta\right) \sqrt{2 T \ln\left(\frac{T^2}{\delta}\right)}- \frac{M}{2 \rho^2 \eta} \right] \nonumber \\
		\frac{c_2 - c_1}{2 \rho M \eta } \leq \frac{1}{c_1 - 1} \left[ 2\left( \frac{c_2}{\rho} + M\eta\right) \sqrt{2 T \ln\left(\frac{T^2}{\delta}\right)}- \frac{M}{2 \rho^2 \eta} \right], 
	\end{gather}
	
	where the last inequality holds by Inequality~\eqref{dif}. By substituting the constants $c_1 = 2$ and $c_2 = 14M$, we get:
	\begin{gather}
		\frac{14M - 2}{2 \rho M \eta } \leq 2\left( \frac{14M}{\rho} + M\eta\right) \sqrt{2 T \ln\left(\frac{T^2}{\delta}\right)} - \frac{M}{2 \rho^2 \eta} \nonumber \\
		\frac{\rho(14M - 2) + M^2}{2 \rho^2 M \eta} \leq 2\left( \frac{14M}{\rho} + M\eta\right) \sqrt{2 T \ln\left(\frac{T^2}{\delta}\right)} \nonumber \\
		\frac{M^2 + 14 \rho M - 2\rho}{2 \rho^2 M \eta} \leq \frac{2(14M + 1)}{\rho} \sqrt{2 T \ln\left(\frac{T^2}{\delta}\right)} \nonumber,
	\end{gather}
	where we used the standard assumption that $\eta \leq \frac{1}{\rho M}$.
	
	Rearranging, we obtain:
	\begin{align*}
		\eta \geq \frac{M^2 + 14 \rho M - 2\rho}{4 \rho M (14M + 1) \sqrt{2 T \ln\left(\frac{T^2}{\delta}\right)}},
	\end{align*}
	
	which yields a contradiction since, by hypothesis, it holds: 
	\begin{align*}
		\eta \leq \eta_{\text{OGD}} &:= \frac{1}{60  M \sqrt{2 T \ln\left(\frac{T^2}{\delta}\right)}}
		< \frac{M^2 + 14 \rho M - 2\rho}{4 \rho M (14M + 1) \sqrt{2 T \ln\left(\frac{T^2}{\delta}\right)}}
	\end{align*}
	
	Thus, we conclude that the assumption must be false, implying $\| \bm{\lambda}_t \|_1 \leq \frac{c_2}{\rho} = \frac{14M}{\rho}$ for each $t \in [T]$.
	This concludes the proof.
\end{proof} 

\section{Omitted Proofs of Section~\ref{sec:violation}}\label{app:v}
\vt*
\begin{proof}
	Let $\tau \leq T$ be the stopping time at which the budget is exhausted. By design of the algorithm, if playing the candidate action exceeds the budget $\beta_j T$ for any resource $j \in [n]$, the algorithm is forced to play the void action $x_t = \varnothing$, which yields zero consumption ($h_{t,j}(\varnothing) = 0$ for all $j \in [n]$). Thus, the budget constraints are never violated, i.e., it holds $\sum_{t \in [T]} h_{t,j}(x_t) = \sum_{t \in [\tau]} h_{t,j}(x_t) \leq \beta_j T$ for all $j \in [n]$.
	For the general constraints, by definition of the OGD update, it holds:
	\[\lambda_{t+1,i} \geq \lambda_{t,i} + \eta g_{t,i} \quad \forall t \in [\tau],  \forall i \in [m]. \] 
	
	Telescoping the sum for all $t \in [\tau]$ we obtain:
	\begin{align}
		\sum_{t \in [\tau ]}\lambda_{t+1,i} &\geq \sum_{t \in [\tau]} \lambda_{t,i} + \eta \sum_{t \in [\tau]} g_{t,i} \nonumber \\
		\lambda_{\tau+1,i} &\geq \lambda_{1,i} + \eta \sum_{t \in [\tau]} g_{t,i}. \nonumber
	\end{align}
	
	Therefore, we can write:
	\begin{align}
		\sum_{t \in [\tau]} g_{t,i}  &\leq \frac{1}{\eta} \lambda_{\tau+1,i} \nonumber \\ 
		&\leq \frac{1}{\eta} \| \bm{\lambda}_{\tau+1} \|_1 \nonumber \\
		&\leq \frac{1}{\eta} \left( \frac{14M}{\rho} \right) \label{sbl} \\
		&= 60 M \sqrt{2T \ln\left( \frac{T^2}{\delta} \right)} \cdot \left( \frac{14M}{\rho} \right) \nonumber \\ 
		&= 840 \frac{M^2}{\rho} \sqrt{2T \ln\left( \frac{T^2}{\delta} \right)}, \nonumber 
	\end{align} 
	
	where Inequality~\eqref{sbl} holds by Lemma~\ref{selfb} with probability at least $1 - \delta$. Since from round $\tau+1$ the algorithm only plays $\varnothing$, we have: 
	\[V_T  := \max_{i \in [m]} \sum_{t \in [T]}g_{t,i}(x_t) = \max_{i \in [m]} \sum_{t \in [\tau]}g_{t,i}(x_t). \]
	Thus, with probability at least $1 - \delta$ it holds: 
	\[V_T  \leq 840 \frac{M^2}{\rho} \sqrt{2T \ln\left( \frac{T^2}{\delta} \right)}.\] 
	This concludes the proof.
\end{proof}

\section{Omitted Proofs of Section~\ref{sec:stoch}}\label{app:st}
\stre* 
\begin{proof} 
	First, we provide an upper bound on the hindsight stochastic optimum. We define the expected Lagrangian evaluated with respect to the input distribution as: 
	\[\bar{\mathcal{L}}_{f,\bm{\tilde{g}}}(\bm{\lambda} | \mathcal{P}) := \mathbb{E}_{(f, \bm{\tilde{g}}) \sim \mathcal{P}} \left[\max_{x \in \mathcal{X}} \mathcal{L}_{f,\bm{\tilde{g}}}(x, \bm{\lambda}) \right].\]  
	By weak duality, $\text{OPT}(\bm{\gamma}) \leq \max_{x \in \mathcal{X}_t}\sum_{t \in T}\mathcal{L}_{f,\bm{\tilde{g}}}(x, \bm{\lambda})$ for any dual variable $\bm{\lambda} \in \mathbb{R}^{M}_+$. Let $\bm{\bar{\lambda}}_{\tau} := \frac{1}{\tau} \sum_{t \in [\tau]} \bm{\lambda}_t$. We can bound the stochastic optimum as: 
	\begin{align}
		\text{OPT}_{stoc} = \mathbb{E}_{\bm{\gamma}}[\text{OPT}(\bm{\gamma})] &= \frac{\tau}{T}\mathbb{E}_{\bm{\gamma}}[\text{OPT}(\bm{\gamma})] + \frac{T - \tau}{T} \mathbb{E}_{\bm{\gamma}}[\text{OPT}(\bm{\gamma})] \nonumber \\ &\leq \tau \bar{\mathcal{L}}_{f,\bm{\tilde{g}}}(\bm{\bar{\lambda}}_{\tau} | \mathcal{P}) + (T - \tau), \label{opt_bound} \end{align}  
	where we used that rewards are in $[0,1]$. 
	
	Now, we analyze the performance of the algorithm up to the stopping time $\tau$. At each step $t \in [\tau]$, Algorithm~\ref{alg:alg1} maximizes the instantaneous Lagrangian choosing $f_t(x_t) = \max_{x \in \mathcal{X}}\mathcal{L}_{f,\bm{\tilde{g}}}(x, \bm{\lambda}_t) + \langle \bm{\lambda}_t, \bm{\tilde{g}}_t(x) \rangle$. Therefore, we can write the instantaneous reward as: 
	\[f_t(x_t) = \mathcal{L}_{f,\bm{\tilde{g}}}(x_t, \bm{\lambda}_t) + \langle \bm{\lambda}_t, \bm{\tilde{g}}_t(x_t) \rangle.\] 
	
	Let $\Gamma_t = (\gamma_1, ..., \gamma_t)$ denote the history of observed inputs up to time $t$ for each $t \in [\tau]$. Since $\bm{\lambda}_t$ is $\Gamma_{t-1}$-measurable, taking the conditional expectation we obtain: 
	\begin{gather} \mathbb{E} [f_t(x_t) | \Gamma_{t-1}] = \mathbb{E}[\mathcal{L}_{f,\bm{\tilde{g}}}(x_t, \bm{\lambda}_t) | \Gamma_{t-1}] + \mathbb{E}  [\langle \bm{\lambda}_t, \bm{\tilde{g}}_t(x_t) \rangle | \Gamma_{t-1}] \nonumber \\
		\mathbb{E} [f_t(x_t) | \Gamma_{t-1}] = \bar{\mathcal{L}}_{f,\bm{\tilde{g}}}(\bm{\lambda}_t | \mathcal{P}) + \mathbb{E}  [\langle \bm{\lambda}_t, \bm{\tilde{g}}_t(x_t) \rangle | \Gamma_{t-1}]. \label{ma} 
	\end{gather}
	
	To analyze the performance up to the stopping time $\tau$, we consider the martingales \[Z_t = \sum_{s \in [t]} \left( \langle \bm{\lambda}_s, \bm{\tilde{g}}_s(x_s) \rangle - \mathbb{E}[\langle \bm{\lambda}_s, \bm{\tilde{g}}_s(x_s) \rangle | \Gamma_{s-1}] \right)\] and $Q_t = \sum_{s \in [t]} \left( f_s(x_s) - \mathbb{E}[f_s(x_s) | \Gamma_{s-1}] \right)$. Since $\tau$ is a bounded stopping time with respect to $\{\Gamma_t\}_{t=1}^T$, from the Optimal Stopping Theorem we have: 
	\begin{equation}
		\mathbb{E}\left[ \sum_{t \in [\tau]}\langle \bm{\lambda}_t, \bm{\tilde{g}}_t(x_t) \rangle \right] =\mathbb{E}\left[ \sum_{t \in [\tau]} \mathbb{E} [\langle \bm{\lambda}_t, \bm{\tilde{g}}_t(x_t) \rangle | \Gamma_{t-1}] \right] \label{stop1} 
	\end{equation} and 
	\begin{equation}
		\mathbb{E}\left[ \sum_{t \in [\tau]} f_t(x_t) \right] =\mathbb{E}\left[ \sum_{t \in [\tau]} \mathbb{E} [f_t(x_t) | \Gamma_{t-1}] \right]. \label{stop2} \end{equation}
	
	Thus, summing Equation~\eqref{ma} over $t$ and taking the expectation on both sides, we have: 
	\[\mathbb{E}\left[ \sum_{t \in [\tau]} \mathbb{E} [f_t(x_t) | \Gamma_{t-1}]\right] = \mathbb{E}\left[ \sum_{t \in [\tau]}  \bar{\mathcal{L}}_{f,\bm{\tilde{g}}}(\bm{\lambda}_t | \mathcal{P})\right] + \mathbb{E}\left[ \sum_{t \in [\tau]}  \mathbb{E}  [\langle \bm{\lambda}_t, \bm{\tilde{g}}_t(x_t) \rangle | \Gamma_{t-1}] \right] .\]  
	
	From Inequalities~\eqref{stop1}-\eqref{stop2} we can write:    
	\begin{align}
		\mathbb{E}\left[ \sum_{t \in [\tau]}f_t(x_t) \right] &= \mathbb{E}\left[\sum_{t \in [\tau]} \bar{\mathcal{L}}_{f,\bm{\tilde{g}}}(\bm{\lambda}_t | \mathcal{P})\right] + \mathbb{E}\left[ \sum_{t \in [\tau]} \langle \bm{\lambda}_t, \bm{\tilde{g}}_t(x_t) \rangle \right] \nonumber \\
		&\geq \mathbb{E}\left[ \tau \bar{\mathcal{L}}_{f,\bm{\tilde{g}}}(\bm{\bar{\lambda}}_{\tau} | \mathcal{P})\right] + \mathbb{E}\left[ \sum_{t \in [\tau]} \langle \bm{\lambda}_t, \bm{\tilde{g}}_t(x_t) \rangle \right], \nonumber
	\end{align}
	where the inequality holds by convexity of the expected Lagrangian.
	Combining this result with Inequality~\eqref{opt_bound}, we can bound the regret as:   
	\begin{align}
		R_T &= \text{OPT}_{stoc} - \mathbb{E}_{\bm{\gamma}}\left[\sum_{t \in [\tau]} f_t(x_t) \right] \nonumber \\
		& \leq \text{OPT}_{stoc}  - \mathbb{E}\left[ \tau \bar{\mathcal{L}}_{f,\bm{\tilde{g}}}(\bm{\bar{\lambda}}_{\tau} | \mathcal{P})\right] - \mathbb{E}\left[ \sum_{t \in [\tau]} \langle \bm{\lambda}_t, \bm{\tilde{g}}_t(x_t) \rangle \right] \nonumber \\ 
		& \leq \mathbb{E}\left[ (T - \tau)\right]  - \mathbb{E}\left[ \sum_{t \in [\tau]} \langle \bm{\lambda}_t, \bm{\tilde{g}}_t(x_t) \rangle \right] \nonumber \\
		&\leq \mathbb{E}\left[ (T - \tau)\right]  + \mathbb{E} \left[-\sum_{t \in [\tau]} \langle \bm{\mu}, \bm{\tilde{g}}_t(x_t) \rangle + \frac{1}{2\eta} \|\bm{\lambda}_1 - \bm{\mu} \|^2_2 + \frac{\eta}{2} TM \right] \label{d1} \\
		&\leq (T - \tau) -(T - \tau) + \frac{1}{\beta_{\min}} + \frac{1}{2\eta \beta_{\min}^2} +  \frac{\eta}{2} TM \label{rb} \\
		&\leq \frac{1}{\beta_{\min}} + \frac{1}{2\eta \beta_{\min}^2} +  \frac{\eta}{2} TM \nonumber \\ 
		&\leq \frac{1}{\beta_{\min}} + \frac{60 M \sqrt{2 T \ln\left( \frac{T^2}{\delta} \right)}}{2 \beta_{\min}^2} + \frac{\sqrt{T}}{120 \sqrt{2 \ln\left( \frac{T^2}{\delta} \right)}} \nonumber, 
	\end{align} 
	where Inequality~\eqref{rb} holds by Lemma~\ref{r_d}. This concludes the proof.
\end{proof}

\section{Omitted Proofs of Section~\ref{sec:adv}}\label{app:adv}
\advreg*
\begin{proof}
	Let $(x_t^*)_{t=1}^T$ be the optimal sequence of actions and let $(x_t^\circ)_{t=1}^T$ be the sequence of strictly feasible actions that leads to the definition of $\rho_{\text{adv}}$, meaning $\tilde{g}_{t,i}(x^\circ_t) \leq -\rho_{\text{adv}}$ for all $i \in [M]$. Let $\alpha = \frac{\rho_\textnormal{adv}}{1 + \rho_\textnormal{adv}}$. 
	For each round $t \in [T]$, we define $\xi_t^\diamond$ as the strategy that plays $x_t^*$ with probability $\alpha$ and $x_t^\circ$ with probability $1 - \alpha$.
	By construction, the sequence $(\xi_t^\diamond)_{t=1}^T$ is always feasible in expectation. Indeed, for each $i \in [M]$, it holds: 
	\begin{align}
		\mathbb{E}_{x\sim \xi_t^\diamond}[\tilde{g}_{t,i}(x)] = \alpha \cdot \tilde{g}_{t,i}(x_t^*) + (1 - \alpha) \tilde{g}_{t,i}(x_t^\circ) \leq \alpha \cdot 1 + (1 - \alpha)(-\rho_{\text{adv}}) = 0 .\label{alpha_g}
	\end{align}
	Moreover, since rewards are non-negative, it holds: 
	\begin{equation} \mathbb{E}_{x\sim \xi_t^\diamond}[f_t(x)] = \alpha \cdot f_t(x_t^*) + (1 - \alpha) f_t(x_t^\circ) \geq \alpha \cdot f_t(x_t^*). \label{alpha_f} \end{equation}   
	Defining $\tau \in [T]$ as the first time that a resource exceeds the budget, we can decompose the $\alpha$-regret: 
	\begin{align}
		\alpha\text{-}R_T &= \alpha \sum_{t \in [T]} f_t(x_t^*) -\sum_{t \in [\tau]}f_t(x_t) \nonumber\\ 
		&\leq \sum_{t \in [T]}\mathbb{E}_{x\sim\xi_t^\diamond}[f_t(x)] -\sum_{t \in [\tau]}f_t(x_t) \label{exp}\\ 
		&= \sum_{t = \tau+1}^T \mathbb{E}_{x\sim\xi_t^\diamond}[f_t(x)] + \sum_{t \in [\tau]} \mathbb{E}_{x\sim\xi_t^\diamond}[f_t(x)] -\sum_{t \in [\tau]}f_t(x_t) \nonumber\\ 
		&\leq (T - \tau) + \sum_{t \in [\tau]} \left( \mathbb{E}_{x\sim\xi_t^\diamond}[f_t(x)] - f_t(x_t)\right) \nonumber\\ 
		&\leq (T - \tau) + \sum_{t \in [\tau]} \left(\langle \bm{\lambda}_t, \mathbb{E}_{x\sim\xi_{t}^\diamond}[\bm{\tilde{g}}_t(x)] \rangle - \langle \bm{\lambda}_t, \bm{\tilde{g}}_t(x_t)\rangle\right) \label{reg_dec} \\
		&\leq (T - \tau) - \sum_{t \in [\tau]} \langle \bm{\lambda}_t, \bm{\tilde{g}}_t(x_t)\rangle,  \label{optA}
	\end{align}
	where Inequality~\eqref{exp} follows from Inequality~\eqref{alpha_f}, Inequality~\eqref{reg_dec} holds since $x_t$ deterministically maximizes the instantaneous Lagrangian over $\mathcal{X}$
	and Inequality~\eqref{optA} follows from Inequality~\eqref{alpha_g}. 
	From Lemma~\ref{r_d}, similarly to the stochastic case, we have: 
	\begin{align*}
		\alpha\text{-}R_T&\leq (T - \tau) -(T - \tau) + \frac{1}{\beta_{\min}} + \frac{1}{2\eta \beta_{\min}^2} +  \frac{\eta}{2} TM \label{rb} \\
		&\leq \frac{1}{\beta_{\min}} + \frac{1}{2\eta \beta_{\min}^2} +  \frac{\eta}{2} TM \nonumber \\ 
		&\leq \frac{1}{\beta_{\min}} + \frac{60 M \sqrt{2 T \ln\left( \frac{T^2}{\delta} \right)}}{2 \beta_{\min}^2} + \frac{\sqrt{T}}{120 \sqrt{2 \ln\left( \frac{T^2}{\delta} \right)}} \nonumber, 
	\end{align*} 
	which concludes the proof.
\end{proof}

\end{document}